\documentclass[a4paper,12pt]{article}
\usepackage{graphicx}
\usepackage{amsfonts,amsmath,amsthm,amssymb}
\usepackage{multirow}
\usepackage{booktabs}
\usepackage{tikz}
\usepackage{fullpage}
\usepackage{mathtools} 
\usepackage{soul}
\usepackage{float}
\usepackage{verbatim}
\usepackage{ulem}

\newcommand{\A}{\mathcal{A}}

\newcommand{\N}{\mathbb{N}}
\newcommand{\R}{\mathbb{R}}

\newtheorem{proposition}{Proposition}

\newtheorem{step}{Step}

\def\p{\smallskip}

\def\A{\boldsymbol{A}}

\def\N{\mathbb{N}}
\def\R{\mathbb{R}}
\def\Z{\mathbb{Z}}

\def\frac#1#2{{#1\over#2}}

\def\boxit#1{\vbox{\hrule\hbox{\vrule\kern.75truemm
\vbox{\kern.75truemm#1\kern1truemm}\kern1truemm\vrule}\hrule}}
\def\epsilon{\varepsilon}

\def\boxit#1{\vbox{\hrule\hbox{\vrule\kern.75truemm
\vbox{\kern.75truemm#1\kern1truemm}\kern1truemm\vrule}\hrule}}
\def\one{\boldsymbol{1}}

\renewcommand{\implies}{\Rightarrow}
\def\core{\mbox{\rm core\hspace{1pt}}}

\numberwithin{equation}{section}

\def\RR{{\mathbb{R}}}

\def\Q{{\mathbb{Q}}}

\def\A{{\cal A}}

\newtheorem{mydef}{Definition}

\newtheorem{myth}{Theorem}

\newtheorem{lemma}{Lemma}

\newtheorem*{ex1cont*}{Example 1 - cont}

\begin{document}

\title{Put-Call Parities, absence of arbitrage opportunities and non-linear pricing rules\thanks{%
\textit{Contacts}: \textit{lorenzo.bastianello@u-paris2.fr},  \textit{chateaun@univ-paris1.fr}, \textit{cornet@ku.edu}. 
 We wish to thank Simone Cerreia-Vioglio for very helpful discussions as well as audiences at Lemma, Paris 1, TUS-VII 2021 (PSE), RUD 2021 (Minnesota).}
}
\author{Lorenzo Bastianello$^a$,
Alain Chateauneuf$^{b,c}$, Bernard Cornet$^{d}$ \\
$^{a}${\scriptsize Universit\'e Paris 2 Panth\'eon-Assas, LEMMA, Paris, France}\\
$^{b}${\scriptsize IPAG Business School, Paris, France}\\
$^{c}${\scriptsize Paris School of Economics and Université Paris 1, Paris, France}\\
$^{d}${\scriptsize University of Kansas, Lawrence, Kansas, USA}
}

\date{
{\large {March 2022}}
}
\maketitle

\begin{abstract}

If prices of assets traded in a financial market are determined by non-linear pricing rules, different versions of the Call-Put Parity have been considered. We show that, under monotonicity, parities between call and put options and discount certificates characterize ambiguity-sensitive  (Choquet and/or \v Sipo\v s) pricing rules, i.e., pricing rules  that can be represented via discounted expectations with respect to non-additive probability measures.  We analyze how non-additivity relates to arbitrage opportunities and we give necessary and sufficient conditions for Choquet and \v Sipo\v s pricing rules to be arbitrage free. Finally, we identify violations of the Call-Put Parity with the presence of bid-ask spreads.

\medskip

\par\noindent
{\sc Keywords:\/} Asset Pricing, Choquet and/or \v Sipo\v s  pricing, No~Arbitrage, Put-Call Parity, Call-Put Parity, Discount Certificate–Call Parity, Market Frictions.
\medskip\par\noindent
{\sc JEL Classification Number:\/} G12, D81, C71.

\end{abstract}

\section{Introduction}\label{sec:intro}

The aim of asset pricing models is to  associate to marketed securities prices that are consistent with the absence of arbitrage opportunities. Typically, the models assume a reference probability over a state space and frictionless pricing rules, as it is done, for instance, in the seminal paper of Harrison and Kreps \cite{HK}. In this case, assets are valued by a linear pricing rule, or, equivalently, as the discounted expectation with respect to the so-called risk-neutral probability. However, both assumptions are questionable. Since at least Ellsberg \cite{Ellsberg}, the decision theory literature suggests a switch to (non-probabilistic) uncertainty. A general framework without an objective probability was already studied in Kreps \cite{Kreps81} and was studied more recently by Riedel \cite{Riedel15}, Cassese \cite{Cassese08}, \cite{Cassese17} and Burzoni, Riedel and Soner \cite{BRS21}. Moreover, important evidence of the presence of frictions in financial markets (transaction costs, taxes, bid-ask spreads) has led to the study of non-linear pricing rules that can account for frictions, as in Garman and Ohlson \cite{GO81}, Prisman \cite{Pris86}, Ross \cite{Ross87}, Bensaid, Lesne, Pages and Scheinkman \cite{ben92} and  Jouini and Kallal \cite{JK95}.

This article drops both assumptions of the existence of an objective probability  and of the linearity of pricing rules. More precisely, following the ideas first developed in  Chateauneuf, Kast and Lapied \cite{CKL}, we consider pricing rules that are Choquet integrals with respect to a \textit{non-additive} risk-neutral probability, i.e., Choquet pricing rules. 
In a recent paper, Cerreia-Vioglio, Maccheroni and Marinacci \cite{CMM} studied financial markets (without a reference probability and with frictions) in which a version of the parity between call and put options is maintained.  They show that, under some basic assumptions, pricing rules satisfying  that parity are precisely Choquet pricing rules. This result is surprising as it connects these two apparently unrelated notions.

However, when prices are non-linear, as for Choquet pricing rules, there are several ways to define the price parity between call and put options. Moreover, without further assumptions on the non-additive risk-neutral probability, arbitrage opportunities may arise. The parity between put and call options is a fundamental non-arbitrage relation in finance, originally derived in Stoll \cite{Stoll}. Call options (denoted here $c^+$ ) and put options ($c^-$) are related by the simple mathematical fact that a real-valued function can be written as the difference between its positive and negative parts: $c=c^+-c^-$. Loosely speaking,  this equality means that an agent who wants to obtain $c$ can buy $c^+$ and sell $c^-$. Obviously, we could write as well $c^-=c^+-c$, and one can get $c^-$ buying $c^+$ and selling $c$. These two different replication strategies were already described  in Stoll \cite{Stoll}, see also Klemkosky and Resnick \cite{KR79}.  

When prices are obtained through non-linear pricing rules, the two non-arbitrage strategies yield two different price parities. The former parity is considered by Cerreia-Vioglio \textit{et al.} \cite{CMM} and it is called  Put-Call Parity (PCP). The latter was analyzed by Chateauneuf  \textit{et al.} \cite{CKL}, and  led to the Call-Put Parity (CPP). 
This paper contributes to the literature of (non-linear) pricing by arbitrage in several ways.  The first part of our paper, Section \ref{sec:parities}, focuses on  put-call parities and how they characterize  non-linear pricing rules. The second part of the paper, Section \ref{sec:NoA}, imposes a non-arbitrage condition and analyzes its implications on non-linear pricing rules.

 Our first contribution in Section \ref{sec:parities} provides the formal relationship between the two parities.  We show that CPP of Chateauneuf \textit{et al.} \cite{CKL} is a stronger assumption: it corresponds to PCP of Cerreia-Vioglio \textit{et al.} \cite{CMM}  plus the absence of bid-ask spreads. Second, we improve the main result of Cerreia-Vioglio \textit{et al.} \cite{CMM}. They proved that a pricing rule satisfies PCP, monotonicity and translation invariance if and only if it is a Choquet pricing rule. Theorem \ref{th:PCP_Choquet} shows that PCP and monotonicity are enough for the characterization. Thus we demonstrate that the connection between PCP and the Choquet pricing rule runs even deeper than suggested by the result of Cerreia-Vioglio \textit{et al.}~\cite{CMM}. Third, we study what happens when CPP replaces PCP. Theorem \ref{th:CPP_Choquet-Sipos} characterizes CPP by a stronger version of Choquet pricing rules, called Choquet-\v Sipo\v s  pricing rules. These are pricing rules that are at the same time Choquet and \v Sipo\v s integrals. The \v Sipo\v s integral is an integral with respect to a non-additive measure that, in general, can differ from the Choquet integral, see \v Sipo\v s  \cite{Sipos}. However when CPP holds, the two  coincide.

The last result of Section \ref{sec:parities} answers  the following natural question.  Is there a financial relation that characterizes \v Sipo\v s pricing rules that are not at the same time Choquet pricing rules? Discount certificates are well known financial options that pays the minimum between the value of an underlying asset and a fixed cap $k$.  We show in Theorem \ref{th:DCP-Sipos} that the two key ingredients to obtain \v Sipo\v s pricing rules are $(i)$ the parity between discount certificates and call options, and $(ii)$ the absence of bid-ask spreads.  To the best of our knowledge, our is the first paper to provide a natural justification for the use of \v Sipo\v s pricing rules for  asset pricing. 

In Section \ref{sec:NoA}, we assume that pricing rules are either Choquet or \v Sipo\v s pricing rules, and we analyze what happens when one imposes the absence of arbitrage opportunities. Intuitively, an arbitrage opportunity is an asset that yields positive payoff and that can be bought at zero or negative cost.
The parity between put and call options is one of the simplest and best understood no-arbitrage relations. However, in general, neither Choquet nor \v Sipo\v s pricing rules  guarantee that markets are arbitrage-free. We study which additional conditions one has to impose in order to eliminate arbitrage opportunities.

 Let $f$ be a Choquet pricing rule used by a dealer to set asset prices in the market. The main result of this section,  Theorem \ref{th:AF_Choquet}  shows that there are no arbitrage opportunities if and only if there exists an  additive risk-neutral measure $\mu$ which is smaller than the  risk-neutral capacity $v$ associated to $f$. The existence of an additive measure $\mu$ below $v$ generalizes the fundamental theorem of asset pricing obtained in the absence of frictions.  On the other hand,  Theorem \ref{th:AF_Choquet} shows that a \v Sipo\v s pricing rule is arbitrage free if and only if it is linear.  This implies that whenever markets do not allow arbitrage opportunities but bid-ask spread are observed, the strong parity CPP of Chateauneuf \textit{et al.} \cite{CKL} must be violated.  
 
 This violation of CPP, in which the marketed price of put options is cheaper than the theoretical one (i.e., the price of building a portfolio that replicates a put), was empirically observed when puts were introduced in  financial markets, see Klemkosky and Resnick \cite{KR79} and  Cremers and Weinbaum~\cite{CW10}. However, it is important to note that  this violation of CPP is consistent with the absence of arbitrage opportunities (and with PCP of Cerreia-Vioglio \textit{et al.}~\cite{CMM}). 

The rest of the paper is organized as follows. Section \ref{sec:framework} introduces the framework and our notation.  Section \ref{sec:parities} studies  call-put parities and the discount certificate-call parity. This section characterizes Choquet, Choquet-\v Sipo\v s and \v Sipo\v s pricing rules. Section \ref{sec:NoA} investigates arbitrage-free Choquet and \v Sipo\v s pricing rules. Section \ref{sec:conclusion} concludes. All proofs are gathered in the Appendix.

\section{The Model}\label{sec:framework}

This paper considers the simple framework of a stochastic two-date model: $t = 0$ (today) is known, $t = 1$ (tomorrow) is uncertain. Uncertainty is represented by a set $\Omega$ (finite or infinite) of \textit{states of nature}, endowed with a $\sigma$-algebra $\A$. One, and only one, state of nature will be realized tomorrow and will be known. Note that we do not assume that there is a reference probability defined on $(\Omega,\A)$.

 A  \textit{payoff}, or  \textit{contingent claim}, is a bounded, real-valued, $\A$-measurable random variable $x:\Omega\rightarrow \R$ (or, if $\Omega$ is finite, a vector $x := (x(\omega))_{\omega\in \Omega}\in \R^{\Omega}$) with $x(\omega)$ representing the payoff (money) at $t = 1$ if state $\omega$ prevails. We denote  $B(\Omega,\A)$ the set of all contingent claims. We adopt the convention that if $x(\omega) < 0$ then $|x(\omega)|$  is paid by the agent and if  $x(\omega) >0$ then $|x(\omega)|$ is received. 
For every $A\in\A$ we denote by $ \one_A$ the payoff in $ B(\Omega,\A)$ defined by $ \one_A(\omega) = 1$ if $\omega\in A$ and $ \one_A(\omega) = 0$ otherwise. We set $ \one_\emptyset = 0$. Note that $B(\Omega,\A)$ comes equipped with the usual partial order, $\geq$, where  $x\geq y$ if and only if $x(\omega)\geq y(\omega)$ for all $\omega\in \Omega$. 
Let $x\in  B(\Omega,\A)$, we denote   $x^+:=x\vee 0$ its \textit{positive part} and $x^-:=(-x)\vee 0$ its \textit{negative part}. Finally we say that two contingent claims $x,y\in B(\Omega,\A)$ are \textit{comonotonic} if   $(x(\omega)-x(\omega'))(y(\omega)-y(\omega'))\geq 0$ for all $\omega,\omega'\in  \Omega$.
%

A \textit{pricing rule} 
is a non-zero function $f: B(\Omega,\A)\rightarrow \R$ that associates to every payoff $x\in B(\Omega,\A)$ the price/cost $f(x)$ to be paid today for the
delivery of the random payoff $x$ at $t = 1$.  Given $x\in B(\Omega,\A)$,  $|f(x)|$ is paid if $f(x)>0$ and  received if  $f(x)<0$ hence $f(x)$ is the buying (ask) price, the price one pays to buy $x$, and the (bid) selling price is then  $-f(-x)$ which is the amount received if one sells $x$.
This paper will take the pricing rule $f: B(\Omega,\A)\rightarrow \R$ as a primitive concept. A pricing rule $f$ is
\begin{itemize}
\item \textit{Monotonic:} : $f(x)\geq f(x')$  for all $x, x'\in B(\Omega,\A)$ such that $x\geq x'$.
\item \textit{Positively Homogeneous:}   $f(tx)=tf(x)$ for all $x\in  B(\Omega,\A)$ and all $t\geq0$.
\item \textit{Translation Invariant:} : $f(x+t \one_\Omega)=f(x)+f(t \one_\Omega)$ for all $x\in  B(\Omega,\A)$ and all $t\in\R_+$.
\item \textit{Subadditive:}    $f(x+x')\leq f(x)+ f(x')$ for all $x, x'\in B(\Omega,\A)$.
\item \textit{Linear:} $f(\alpha x+\beta x')= \alpha f(x)+\beta f(x')$ for all $x, x'\in B(\Omega,\A)$, $\alpha,\beta\in\R$.
\item \textit{Comonotonicly Additive:}   $f(x+x')= f(x)+ f(x')$ for all $x, x'\in B(\Omega,\A)$ comonotonic.
\item \textit{No Bid-Ask Spread:}   $f(x)=-f(-x)$ for all $x\in B(\Omega,\A)$.
\item \textit{Non-negative Bid-Ask Spread:}   $f(x)\ge -f(-x)$ for all $x\in B(\Omega,\A)$.
\item \textit{Non-negative:}  $x\geq 0$ implies $f(x)\geq 0$. 
\end{itemize}

If $f(x)>-f(-x)$ then the price for buying $x$ (bid price) in the market is higher than the price for selling $x$ (ask price). In this case there is a \textit{bid-ask spread}. One can also have a bid-ask spread when 
$f(x)<-f(-x)$ but this situation will not prevail since it generates an arbitrage opportunity (see Section \ref{sec:NoA}). 

A \textit{capacity}  $v$ on the  measurable space $(\Omega,\A)$ is a set function $v:\A \mapsto \mathbb{R}$ such that $v(\emptyset)=0$ and which is monotone, that is, for all $A,B \in \A,\, A \subseteq B$ implies $v(A)\leq v(B)$. The capacity   is said to be normalized, also called   \textit{non-additive probability}, if $v(\Omega)=1$.  A capacity $v: \A\mapsto \mathbb{R}$ is \textit{concave,} also called submodular, if, for all $A,B\in \A$, $v(A\cup B)+v(A\cap B)\leq v(A)+v(B)$.  
A \textit{probability}  $\mu:\A \mapsto \mathbb{R}$ is a normalized capacity which is (finitely) additive, that is,    $A\cap B=\emptyset$ implies $\mu(A\cup B)=\mu(A)+\mu(B)$. 
The conjugate of the  capacity  $v:\A \mapsto \mathbb{R}$ is the capacity $v^*:\A \mapsto \mathbb{R}$ defined by
 $v^*(A):=  v(\Omega)-v(A^c)$ for all $A\in \A$ and we say that the capacity $v$ is \textit{auto-conjugate} if $v=v^*$. Note that a capacity is auto-conjugate if and only if $v(A)+v(A^c)=v(\Omega)$ for all $A\in \A$.

Consider a capacity $v$ on $(\Omega, \A)$. The function $f$ is said to be a \textit{Choquet pricing rule}  with respect to $v$ if
$$
f(x)=\int_{\Omega}^C{x\,dv}:=\int_{-\infty}^0{(v(\{x\geq t\})-v(\Omega))\,dt} + \int_0^{+\infty}{v(\{x\geq t\})\,dt}  \text{ for all }   x\in  B(\Omega,\A),
$$
where $\{x\geq t\}=\{\omega\in\Omega | x(\omega)\geq t\}$. The notation $\int^C$ indicates that the integral is a Choquet integral (we drop subscript $\Omega$ from the integral when no confusion arises).  We will use the notation $\int$ for the standard integral with respect to an additive capacity. 
We note that the capacity $v$  associated with the Choquet pricing rule $f$ is uniquely defined, it  satisfies $ v(A):=f(\one_A)$ for all $A\in \A$, and  $v$ is   called the risk-neutral capacity (associated with  the Choquet pricing rule $f$). 

We say that  $f$ is a \textit{\v Sipo\v s  pricing rule} with respect to $v$ if
$$
f(x)=\int_{\Omega}^S{x\,dv}:= \int^C{x^+\,dv}-\int^C{x^-\,dv} \; \text{ for all }   x\in  B(\Omega,\A).
$$

Note that  Choquet and \v Sipo\v s integrals coincide if $x\in  B(\Omega,\A)$ is non-negative. One fundamental difference between the two integrals is that the former satisfies translation invariance, while the latter has no bid-ask spread.

Finally, we say that $f$ is a \textit{Choquet-\v Sipo\v s pricing rule} when the Choquet  and \v Sipo\v s integrals coincide with $f$ for the same capacity $v$, that is
$$
f(x)=\int_{\Omega}^S{x\,dv} = \int_{\Omega}^C{x\,dv}  \; \text{ for all }   x\in  B(\Omega,\A).
$$

We recall (Proposition \ref{prop:Sipos-charact}) that $f$ is a \textit{Choquet-\v Sipo\v s pricing rule} if and only if $f$ is a \textit{Choquet pricing rule} with respect to an auto-conjugate capacity $v$.

We end this section with a remark. Suppose that $f$ is a Choquet pricing rule with respect to $v$.  We can normalize the capacity $v$, by defining the non-additive probability $\bar v:\A \to \RR$ by $  \bar v(A):= \frac{v(A)}{v(\Omega)} $ for all $A\in \A$, and  obtain a riskless interest rate  $r>-1$ uniquely defined by   $f(\one_\Omega)=\frac{1}{1+r}$.\footnote{Note that $v(\Omega)=f(\one_{\Omega})>0$ otherwise one would obtain the contradiction $f(x)=0$ for all $x\in  B(\Omega,\A)$.} Then a Choquet pricing rule can be written as
$$
f(x)=\frac{1}{1+r}\int^C x \,d\bar v, \, \forall x\in  B(\Omega,\A).
$$
 This (non-additive) probabilistic formulation,  taken from  Cerreia-Vioglio \textit{et al.} \cite{CMM}, allows to interpret the cost $f(x)$ of every payoff $x$ as the present value of its non-additive expectation, where the present value is calculated with the riskless interest rate. The same can be done for \v Sipo\v s and Choquet-\v Sipo\v s pricing rules.


\section{Characterization of pricing rules through financial parities}\label{sec:parities}

\subsection{Call-Put Parity(ies)}\label{sec:CPP}

This section analyzes the implications on pricing rules of the two different parities between prices of call and put options introduced by Chateauneuf \textit{et al.} \cite{CKL} and Cerreia-Vioglio \textit{et al.} \cite{CMM}. 

A call option with strike $k$ and expiration date $T$ is a financial contract that gives the option buyer the right, but not the obligation, to buy a stock, bond, or commodity at price $k$ at time $T$.  The stock, bond, or commodity is called the underlying asset. A call buyer profits when the underlying asset increases in price.
A put option gives the right, but not the obligation, to sell the underlying asset at price $k$ at time $T$. 

In our two-period financial economy, given an underlying asset $x\in B(\Omega,\A)$, a call option with strike $k\geq 0$ is defined by $c_{x,k}=(x-k\one_\Omega)^+$ and the related put option is defined by $p_{x,k}=(k\one_\Omega-x)^+$. The following equality holds:
\begin{equation}\label{eq:PCP_math}
x=c_{x,k}-p_{x,k}+k\one_\Omega,
\end{equation}
and says that one can replicate the underlying $x$ by buying a call, selling a put with strikes $k$, and buying $k$ units of the bond. Given a pricing rule $f$, Cerreia-Vioglio \textit{et al.} \cite{CMM} used equality in (\ref{eq:PCP_math}) in order to derive the Put-Call Parity (PCP)
\begin{equation}\label{eq:PCP}
\tag{PCP}
f(x)= f(c_{x,k})+f(-p_{x,k})+f(k\one_\Omega).
\end{equation}

Obviously, the mathematical equality in (\ref{eq:PCP_math}) can be re-written as 
\begin{equation}\label{eq:CCP_math}
p_{x,k}=c_{x,k}-x+k\one_\Omega.
\end{equation}
Hence a put with strike $k$ can be replicated  buying a call with the same strike, selling the underlying asset, and buying  $k$ units of the bond. Chateauneuf \textit{et al.} \cite{CKL} used equation (\ref{eq:CCP_math}) to derive a different parity that they named Call-Put Parity (CPP)
\begin{equation}\label{eq:CPP}
\tag{CPP}
f(p_{x,k})=f(c_{x,k})+f(-x)+f(k\one_\Omega).
\end{equation}
If $f$ is linear, \ref{eq:PCP} and \ref{eq:CPP} are equivalent. The following result provides the formal relation between the two parities. It shows that  \ref{eq:CPP} of  Chateauneuf \textit{et al.} \cite{CKL} is stronger than  \ref{eq:PCP} of  Cerreia-Vioglio \textit{et al.} \cite{CMM} as   \ref{eq:CPP} is equivalent to \ref{eq:PCP} and the absence of bid-ask spreads.

\begin{proposition}\label{prop:CCP-PCPandNS}
Let $ f :  B(\Omega,\A)\rightarrow \R$ be a pricing rule.  Then $(i)\Leftrightarrow(ii)$.
\begin{itemize}
\item[(i)] $f$ satisfies \ref{eq:CPP};
\item[(ii)] $f$ satisfies \ref{eq:PCP} and no bid-ask spreads.
\end{itemize}
\end{proposition}
\begin{proof}
The proof of Proposition \ref{prop:CCP-PCPandNS} is given in Appendix \ref{proof:CCP-PCPandNS}.
\end{proof}

When frictions are taken into account, linearity of $f$ is no longer guaranteed. Choquet pricing rules are examples of non-linear pricing rules that can explain market frictions. They have been considered first  by Chateauneuf \textit{et al.} \cite{CKL}, see also Wang, Young and Panjer \cite{Wang97} and Castagnoli, Maccheroni and Marinacci \cite{CMM02}, Chateauneuf and Cornet~\cite{CCMP},~\cite{CCET}.
A full characterization using a financial parity has been given  by Cerreia-Vioglio \textit{et al.} \cite{CMM}. The main theorem of Cerreia-Vioglio \textit{et al.} \cite{CMM} proves that a pricing rule $f$ satisfies \ref{eq:PCP}, monotonicity and translation invariance if and only if it is a Choquet pricing rule. Their idea is extremely interesting as it connects two apparently unrelated concepts: the well known financial concept of \ref{eq:PCP}  with the Choquet integral.
 We generalize their result showing that translation invariance is  redundant. 

\begin{myth}\label{th:PCP_Choquet}
Let $ f :  B(\Omega,\A)\rightarrow \R$ be a monotone pricing rule. Then $(i)\Leftrightarrow(ii)$.
\begin{itemize}
\item[(i)] $f$ satisfies \ref{eq:PCP};
\item[(ii)] $f$ is a Choquet pricing rule.
\end{itemize}
\end{myth}
\begin{proof}
The proof of Theorem \ref{th:PCP_Choquet} is given in Appendix \ref{proof:PCP_Choquet}.
\end{proof}

Note that \textit{only} monotonicity and \ref{eq:PCP} are required to pin down Choquet pricing rules. These are very weak and desirable assumptions. This demonstrates the central role played by Choquet pricing rules. From a mathematical point of view, our proof differs from the one of Cerreia-Vioglio \textit{et al.} \cite{CMM}.  Their characterization   is based on a representation result by Greco \cite{Greco}.  On the other hand, our  proof of Theorem \ref{th:PCP_Choquet}    consists in showing that the parity  \ref{eq:PCP} is equivalent to the comonotonic additivity of the pricing rule $f$, a familiar concept in Decision Theory. This in turns  allows us to use the fundamental paper of Schmeidler \cite{Schmeidler86}  characterizing Choquet pricing rules by the comonotonic additivity property. 

 A natural question arises: what happens when one replaces \ref{eq:PCP} by \ref{eq:CPP}? Theorem~\ref{th:CPP_Choquet-Sipos} shows that  \ref{eq:CPP} (together with  monotonicity)  characterizes Choquet-\v Sipo\v s pricing rules. Before stating the theorem, we give a characterization of Choquet-\v Sipo\v s pricing rules.

 \begin{proposition}\label{prop:Sipos-charact}
Let $ f :  B(\Omega,\A)\rightarrow \R$ be a Choquet pricing rule with respect to $v$.  Then the following are equivalent.
\begin{itemize}
\item[(i)]  $f(x)+f(-x)= 0$ for all $x\in  B(\Omega,\A)$.
\item[(ii)] $v=v^*$, i.e., $v(A)+v(A^c)=v(\Omega)$ for all $A\in \A$.
\item[(iii)] $f$ is a Choquet-\v Sipo\v s pricing rule. 
\end{itemize}
\end{proposition}
\begin{proof}
The proof of Proposition \ref{prop:CCP-PCPandNS} is given in Appendix \ref{proof:CPP_Choquet-Sipos}.
\end{proof}
 
Theorem \ref{th:CPP_Choquet-Sipos} below shows that a monotone pricing rule satisfies \ref{eq:CPP} of Chateauneuf \textit{et al.} \cite{CKL} if and only if it is a Choquet-\v Sipo\v s pricing rule.

\begin{myth}\label{th:CPP_Choquet-Sipos}
Let $ f :  B(\Omega,\A)\rightarrow \R$ be a  monotone pricing rule.  Then $(i)\Leftrightarrow(ii)$.
\begin{itemize}
\item[(i)] $f$ satisfies \ref{eq:CPP};
\item[(ii)] $f$ is a Choquet-\v Sipo\v s pricing rule.
\end{itemize}
\end{myth}
\begin{proof}
The proof of Theorem \ref{th:CPP_Choquet-Sipos} is given in Appendix \ref{proof:CPP_Choquet-Sipos}.
\end{proof}

Note that if assets are priced through Choquet-\v Sipo\v s pricing rules, then there are no bid-ask spreads. This was also noted in a recent paper of L\'ecuyer and Lefort \cite{LL20} who studied particular Choquet pricing rules given by (normalized) generalized neo-additive capacities (GNAC) (see Chateauneuf, Eichberger and Grant \cite{neo}). They show that if $f$ is a Choquet pricing rule, then there are no bid-ask spreads if and only if $v(A)+v(A^c)=v(\Omega)$ for all $A\in  \A$ (i.e., $f$ is Choquet-\v Sipo\v s). Moreover if $f$ is a GNAC pricing rule,  there are no bid-ask spreads if and only if $v(A)+v(A^c)=v(\Omega)$ for at least one $A\in  \A$. See also Castagnoli, Maccheroni and Marinacci \cite{CMM04}.

To summarize, Section \ref{sec:CPP} characterizes Choquet pricing rules through \ref{eq:PCP} and Choquet-\v Sipo\v s pricing rules through \ref{eq:CPP}. Next section characterizes general \v Sipo\v s pricing rules  through the parity between discount certificates and call options.

\subsection{Discount Certificate-Call Parity }\label{sec:DCP}

In order to pin down \v Sipo\v s pricing rules we have to introduce discount certificates. A \textit{discount certificate} on an asset $x$ with cap $k\geq 0$, denoted $d_{x,k}$, is a contingent claim that in state $\omega$ pays $x$ if $x(\omega)\leq k$ and $k$ if $x(\omega)>k$, or equivalently,
$$
d_{x,k}= x\wedge k\one_\Omega.
$$
Noting that $d_{x,k}=(x- k\one_\Omega) \wedge0+ k\one_\Omega$ and recalling that  $c_{x,k}:= [x - k \one_\Omega]^+ =(x-k\one_\Omega)\vee  0$ one can conclude
\begin{equation}\label{eq:DCP_math}
x=  c_{x,k} +d_{x,k}.
\end{equation}
Therefore asset $x$ can be replicated buying a call and a discount certificate. Given a pricing rule $f$, Cerreia-Vioglio \textit{et al.} \cite{CMM} used equation (\ref{eq:DCP_math}) to derive the Discount Certificate-Call Parity (DCP)
\begin{equation}\label{eq:DCP}
\tag{DCP}
f(x)=  f(c_{x,k}) +f(d_{x,k}).
\end{equation}
In their paper they used \ref{eq:DCP}, monotonicity and translation invariance in order to derive another characterization of Choquet pricing rules. In Theorem \ref{th:DCP-Sipos} we show that \ref{eq:DCP} and no bid-ask spreads pin down  \v Sipo\v s pricing rules.
 
\begin{myth}\label{th:DCP-Sipos}
Let $ f :  B(\Omega,\A)\rightarrow \R$ be a  monotone pricing rule.  Then $(i)\Leftrightarrow(ii)$.
\begin{itemize}
\item[(i)] $f$ satisfies no bid-ask spreads and \ref{eq:DCP};
\item[(ii)] $f$ is a \v Sipo\v s pricing rule.
\end{itemize}
\end{myth}
\begin{proof}
The proof of Theorem \ref{th:DCP-Sipos} is given in Appendix \ref{proof:DCP-Sipos}.
\end{proof}

As we noted in Section \ref{sec:CPP}, when pricing rules are non-linear it is important to pay attention to the replication strategy, as different strategies imply different parities.  The same is true with the  Discount Certificate-Call Parity. In fact equation (\ref{eq:DCP_math}) can be rewritten as $c_{x,k}=x-d_{x,k}$. Therefore one can replicate a call by buying the underlying $x$ and selling a discount certificate. This replication strategy would suggest the parity
\begin{equation}\label{eq:DCP*}
\tag{DCP*}
f(c_{x,k})=f(x)+f(-d_{x,k}),
\end{equation}
which may differ from \ref{eq:DCP} if $f$  is not linear. Proposition \ref{prop:DCP(DCP*} relates \ref{eq:DCP} and \ref{eq:DCP*} in the same way Proposition \ref{prop:CCP-PCPandNS} relates \ref{eq:PCP} and \ref{eq:CPP}.

\begin{proposition}\label{prop:DCP(DCP*}
Let $ f :  B(\Omega,\A)\rightarrow \R$ be a non-zero pricing rule.  Then $(i)\Leftrightarrow(ii)$.
\begin{itemize}
\item[(i)] $f$ satisfies \ref{eq:DCP*};
\item[(ii)] $f$ satisfies \ref{eq:DCP} and no bid-ask spreads.
\end{itemize}
\end{proposition}
\begin{proof}
The proof is similar to the one of Proposition \ref{prop:CCP-PCPandNS} and it is therefore omitted.
\end{proof}

While the parities between call and put options and discount certificate and call options are some of the best known assumptions about the absence of arbitrage opportunities, the pricing rules obtained in Theorem \ref{th:PCP_Choquet}, Theorem \ref{th:CPP_Choquet-Sipos} and Theorem \ref{th:DCP-Sipos} do not guarantee that markets are arbitrage free.  The following section defines arbitrage opportunities and shows how one can eliminate them.


\section{Absence of Arbitrage Opportunities}\label{sec:NoA}

In the characterizations given in  Section \ref{sec:parities}, no property is stated on the capacity $v$ associated to the pricing rule $f$. Without additional properties on $v$, Choquet and \v Sipo\v s pricing rules may leave room for some arbitrage opportunity.  

Intuitively, a payoff $x\in B(\Omega,\A)$ is an arbitrage opportunity if ``it allows to make money from nothing''. Note that whenever $f$ is assumed to be subadditive, the pricing rule $f$ is said to be arbitrage-free whenever $x\ge 0$ implies that $f(x)\ge 0$, that is, there is no payoff   $x\ge0$ (with no loss at each state tomorrow) such that $f(x)<0$   (with a gain today). Without the subadditivity assumption, there is a need to eliminate  other arbitrage opportunities. Consider for the moment \textit{buy $\&$ sell arbitrage opportunities}, that is payoffs $x\in  B(\Omega,\A)$ for which $f(x) +f(-x)<0$ leading to a zero stream of money tomorrow ($x-x=0$) and a gain today.  The  following result characterizes Choquet  pricing rules for which there is no buy $\&$ sell arbitrage opportunity. Finally note that the absence of buy $\&$ sell arbitrage opportunities is equivalent to the property  Non-negative Bid-Ask Spreads, i.e., $f(x)+f(-x)\ge 0$ for all $x\in  B(\Omega,\A)$.\footnote{Note that buy $\&$ sell arbitrage opportunities cannot arise for  \v Sipo\v s and Choquet-\v Sipo\v s pricing rules as these pricing rules satisfy No Bid-Ask Spread.}

\begin{myth}\label{th:CPP_Choquet-Sipos10} 
Let $ f :  B(\Omega,\A)\rightarrow \R$ be a Choquet  pricing rule with respect to $v$.  Then the following are equivalent.
\begin{itemize}
\item[(i)] $f$ has Non-negative Bid-Ask Spreads; 
\item[(ii)] $v\geq v^*$, i.e., $v(A)+v(A^c)\ge v(\Omega)$ for all $A\in \A$;
\item[(iii)]  $-f(-x)\le  \int^{S}x \, dv\le f(x)$ for all $x\in  B(\Omega,\A)$.
\end{itemize}
\end{myth} 
\begin{proof}
The proof of Theorem \ref{th:CPP_Choquet-Sipos10} is given in Appendix \ref{proof:CPP_Choquet-Sipos}.
\end{proof}

Note that the previous theorem, part of which was proved when $\Omega $ is finite by    Chateauneuf and  Cornet \cite{CCETB} introduces the new property that the \v Sipo\v s pricing rule is an homogeneous selection of the bid-ask spread interval, that is,  $\int^{S}x \, dv \in   [-f(-x), f(x)]$ for all $x\in  B(\Omega,\A)$.  Pricing a new security $x\in  B(\Omega,\A)$ so that the enriched pricing rule remains buy $\&$ sell  arbitrage-free amounts to choose the price of the new security in the interval  $[-f(-x), f(x)]$; thus having a canonical way to do it via the \v Sipo\v s pricing rule is both interesting for theoretical reasons but also for practical ones since the calculation is standard.

The notion of buy $\&$ sell arbitrage opportunity can be generalized in the spirit of Chateauneuf and Cornet \cite{CCETB}.

\begin{mydef}\label{def:AF}
$f$ is \textit{Arbitrage Free ($AF$)} if for all $n \in\N$ and all  $x_1,\dots,x_{n }\in  B(\Omega,\A)$  
\begin{equation}\label{eq:AF}
\tag{AF}
\sum_{i=1}^n x_i\ge 0 \,\, \Rightarrow  \,\, \sum_{i=1}^n\ f(x_i)\ge 0.
 \end{equation}
 \end{mydef} 
 
We say equivalently that there are no arbitrage opportunities, that the market is Arbitrage Free (AF) or that a pricing rule $f$ satisfies the AF condition. The interpretation of Definition \ref{def:AF} is the following. Suppose that a customer wants to construct an asset $x\ge 0$, buying $n$ securities $x_1,\dots,x_n$. Then the pricing rule is AF if and only if the cost of forming this portfolio in non-negative.  Thus we can also interpret it as the absence of {\it multiple}  buy $\&$ sell arbitrage opportunities. Note that Definition \ref{def:AF} is stronger than   the standard  definition in which $f$ is called arbitrage free if and only if $x\geq0$ implies $f(x)\geq0$.

Note that when $f$ is subadditive (hence when $f$ is linear too), then it is arbitrage free if and only if it is non-negative (i.e., precisely when $x\geq0$ implies $f(x)\geq0$). The fundamental theorem of asset pricing famously characterizes linear and AF pricing rules as discounted expectation with respect to a (additive) probability, see Harrison and Kreps~\cite{HK}.

When frictions are taken into account, the linearity of $f$ is no longer guaranteed. In general, non-linear pricing rules do not guarantee the absence of arbitrage opportunities. To solve this issue, usually  pricing rules are required to be sublinear, i.e., to satisfy positive homogeneity and subadditivity, see Jouini and Kallal \cite{JK95}, Castagnoli \textit{et al.} \cite{CMM02}, Araujo, Chateauneuf, and Faro \cite{ACF18}, Burzoni \textit{et al.} \cite{BRS21} and Chateauneuf and Cornet~\cite{CCMP}. Choquet pricing rules satisfy positive homogeneity, but, in general, they are not subadditive. Subadditive Choquet pricing rules were  considered by Chateauneuf \textit{et al.} \cite{CKL} and are characterized by concave capacities, see also Chateauneuf and Cornet \cite{CCET},  Cerreia-Vioglio \textit{et al.} \cite{CMM} and Cinfrignini, Petturiti and Vantaggi \cite{Cin21} for the particular case of belief functions. Concavity of the capacity gives a sufficient condition in order to guarantee the absence of arbitrage opportunities for Choquet pricing rules. However the condition  is not necessary.

Theorem \ref{th:AF_Choquet} provides a full characterization of AF for both Choquet and \v Sipo\v s pricing rules. For Choquet  pricing rules there are no arbitrage opportunities if and only if there exists a positive and additive set function $\mu$ below the capacity $v$. Formally, the \textit{anticore} of a capacity $v$ (already considered in Gilboa and Lehrer \cite{GL91} and used in the context of pricing rules by Araujo, Chateauneuf and Faro \cite{ACF12}) is defined as
\begin{equation}\label{eq:AC}
 \tag{AC}
 AC(v)=\{\mu:\A \to \mathbb{R}_+,\,\mu \text{ is   additive, } \mu\leq v, \text{ and } \mu(\Omega)=v(\Omega)\}.
\end{equation}
It is well known that if a capacity $v$ is concave, then the associated Choquet pricing rule $f$ is subadditive and  $AC(v)\neq\emptyset$. Theorem \ref{th:AF_Choquet} shows that if $f$ is a Choquet pricing rule with respect to $v$ then it is AF if and only if $AC(v)\neq\emptyset$. A similar result was proved in  Chateauneuf and Cornet \cite{CCETB} when $\Omega $ is finite. However, the proof for the case of  $\Omega $ infinite is different.

Theorem \ref{th:AF_Choquet} also characterizes  AF \v Sipo\v s pricing rules. As it turns out, a \v Sipo\v s pricing rule is AF if and only if it is linear. Therefore Sipo\v s pricing rules that are also AF  cannot take into account \textit{any} friction.

\begin{myth}\label{th:AF_Choquet}
Let $f :  B(\Omega,\A)\rightarrow\R$ be a pricing rule. Then:
\begin{enumerate}
\item[(i)] If $f$ is a  Choquet  pricing rule, then $f$ satisfies $AF$ if and only if $AC(v)\neq\emptyset$.
\item[(ii)]   If $f$ is a \v Sipo\v s pricing rule, then $f$ satisfies AF  if and only if $f$ is linear.
\end{enumerate}
\end{myth}
\begin{proof}
The proof of Theorem \ref{th:AF_Choquet} is given in Appendix \ref{proof:AF_Choquet}.
\end{proof}

To conclude,   we obtain that under the mild assumption of monotonicity and AF the unique parity which would lead to frictions in accordance with  observed  bid-ask spreads  is  \ref{eq:PCP} considered in  Cerreia-Vioglio \textit{et al.} \cite{CMM}. We remark that under \ref{eq:PCP}, monotone pricing rules are of the Choquet type. If moreover non-negative bid-ask spreads are present in the market, then it turns out that $f(p_{x,k})\leq f(c_{x,k})+f(-x)+f(k\one_\Omega)$ for all $x\in B(\Omega,\A)$, $k\geq 0$, i.e., a violation of  \ref{eq:CPP} considered by Chateauneuf \textit{et al.} \cite{CKL}. This price violation says that it is cheaper to buy a put option with strike $k$ in the market, rather than ``constructing" one by buying a call, selling the underlying $x$ and buying $k$ units of the bond. This miss-pricing was actually observed when put options were introduced in the markets, see Gould and Galai  \cite{GouldGalai}, Klemkosky and Resnick \cite{KR79} and Sternberg \cite{Sternberg}.

\section{Conclusion}\label{sec:conclusion}

The first part of our paper studies  different formulations of the famous put-call parity in a framework where pricing rules are non-linear. Cerreia-Vioglio \textit{et al.} \cite{CMM} study a parity called  Put-Call Parity (PCP). They prove that PCP, together with translation invariance and monotonicity, characterizes Choquet pricing rules. Chateauneuf  \textit{et al.} \cite{CKL} study a  different Call-Put Parity (CPP). Our first result studies the relation between the two parities and shows that CPP is equivalent to PCP and no bid-ask spreads. Our second result improves the characterization of  Cerreia-Vioglio \textit{et al.} \cite{CMM} as it shows that  translation invariance is redundant. Third, replacing PCP by CPP, we obtain Choquet-\v Sipo\v s  pricing rules, which are pricing rules that are at the same time Choquet and \v Sipo\v s pricing rules. Fourth, we characterize \v Sipo\v s  pricing rules using the Discount Certificate-Call parity and no bid-ask spreads.

The second part of our paper studies the implication of an arbitrage free (AF) condition on pricing rules. We show that, in general, Choquet and \v Sipo\v s  pricing rules are not AF. If $f$ is a Choquet pricing rule with respect to a capacity $v$, then the market is AF if and only if the anticore of $v$ is non-empty. If $f$ is a \v Sipo\v s  pricing rule then  $f$ is AF if and only if it is linear. 
Finally, we characterize bid-ask spreads as violations of CPP. This shows that, for pricing rules \textit{\`a la} Choquet, one can observe at the same time PCP, absence of arbitrage opportunities and a violation of CPP (i.e., non-negative bid-ask spreads).


\appendix
\section{Appendix}
\small

\subsection{Summary of the properties of pricing rules}

We gather here the properties of pricing rules used in the proofs and their abbreviations. Also, recall that $c_{x,k}=(x-k \one_\Omega)^+$, $p_{x,k}=(k \one_\Omega-x)^+$ and $d_{x,k}= x\wedge k \one_\Omega$. We will need the following basic equalities:
\begin{equation}\label{eq:+_vee_w}
(x-k \one_\Omega)^+=x\vee k \one_\Omega-k \one_\Omega \text{, } -(k \one_\Omega-x)^+= x\wedge k \one_\Omega-k \one_\Omega, \text{ and } x\wedge k \one_\Omega= (x- k \one_\Omega) \wedge0+ k \one_\Omega.
\end{equation}

\begin{enumerate}
\item Monotonicity (M). $f(x)\geq f(x')$ for all $x\geq x'$.
\item Translation Invariance (TI). $f(x+k \one_\Omega)=f(x)+f(k \one_\Omega)$ for all $x\in  B(\Omega,\A)$ and all $k\in \R_+$.
\item Put-Call Parity (PCP). $f(x)= f((x-k \one_\Omega)^+)+f(-(k \one_\Omega-x)^+)+f(k \one_\Omega)$ for all $x\in  B(\Omega,\A)$ and all $k\in \R_+$.
\item  Call-Put Parity (CPP). $f((k \one_\Omega-x)^+)=f((x-k \one_\Omega)^+)+f(-x)+f(k \one_\Omega)$ for all $x\in  B(\Omega,\A)$ and all $k\in \R_+$.

\item \nonumber $f((k \one_\Omega-x)^+)=f((x-k \one_\Omega)^+)+f(-x)+f(k \one_\Omega)$ for all $x\in  B(\Omega,\A)$ and all $k\in \R_+$.
\item  Discount Certificate - Call Parity  (DCP). $ f(x) = f((x-k \one_\Omega)^+) +f(x\wedge k \one_\Omega)$ for all $x\in B(\Omega,\A)$, all $k\ge 0$.
\item No-Bid Ask Spread (NS).
$f(-x) =-f(x)$ for all $x\in B(\Omega,\A)$. 
\end{enumerate}


\subsection{Proof of Proposition \ref{prop:CCP-PCPandNS}}\label{proof:CCP-PCPandNS}

\begin{proof}
$\bullet$ $[(i) \Rightarrow (ii)]$ We prove NS. Note that $f(0)=0$ is a consequence of CPP (taking $x=0$ and $k=0$). 
Take now $k=0$ in CPP, then $f((-x)^+)= f(x ^+)  + f(-x)$ for all $x\in  B(\Omega,\A)$. Replacing $x$ by $-x$, we  get $f(x^+)= f((-x )^+)  + f(x)$. Thus
$
f(x)=f(x^+)-f((-x )^+)
=-f(-x).
$ Hence, $f$ has no bid-ask spread.
From CPP and NS we then deduce that PCP holds since:

\p\qquad $f((x-k \one_\Omega)^+)+ f(k \one_\Omega)= f(x) -f(-(k \one_\Omega-x)^+)=-f(-x)+f((k \one_\Omega-x)^+).$

$\bullet$ $[(i) \Leftarrow (ii)]$ Immediate. 
\end{proof}

\subsection{Proof of Theorem \ref{th:PCP_Choquet}}\label{proof:PCP_Choquet}

The following lemmata will prove useful in the proof of Theorem \ref{th:PCP_Choquet}.

\begin{lemma}\label{lemma:strong_TI}
A pricing rule $f:B(\Omega,\A)\rightarrow\R$ satisfies Translation Invariance if and only if 

\p\qquad$f(x+k\one_\Omega)=f(x)+kf( \one_\Omega)\; \text{ for all}\; x\in  B(\Omega,\A)\; \text{and all }\; k\in \R.$
\end{lemma} 
\begin{proof} 
$\bullet$ [$\Rightarrow$] We first prove that, $\text{ for all}\; x\in  B(\Omega,\A)\; \text{and all }\; k\in \R , f(x+k\one_\Omega)=f(x)+f( k\one_\Omega)$. Indeed, for $k\in \R_+$ this follows from the translation invariance of $f$. Moreover,  by Translation Invariance of $f$, we have $f(0)=0$   (taking $x=0$ and $k=0$) and  
$0=f(0)= f(-k\one_\Omega+k\one_\Omega)=f(-k\one_\Omega)+f(k\one_\Omega)$ (taking $x=-k\one_\Omega$). Thus 
$f(-k\one_\Omega)=-f(k\one_\Omega)$ for all  $  k\in \R_+ $. Finally,  $f(x)=f(x-k\one_\Omega+k\one_\Omega)=f(x-k\one_\Omega)+f(k\one_\Omega)$, hence:

\p\qquad $f(x-k\one_\Omega)=f(x)-f(k\one_\Omega) =f(x)+ f(-k\one_\Omega)$ for all $ k\in \R_+$.
\p
%

 We complete the proof by showing that $f( k\one_\Omega)=kf( \one_\Omega)$ for all $k\in \R$. 
First, let $n\in \N$ and $t\in\R$, then TI implies $f(nt\one_\Omega)=f((n-1)t\one_\Omega+t\one_\Omega)=f((n-1)t\one_\Omega)+f(t\one_\Omega)$ and by induction $f(nt\one_\Omega)=nf(t\one_\Omega)$. Consequently, from the first part of the proof, $f(-n\one_\Omega)=-nf(\one_\Omega)$ for all $n\in \N$. Second, let $q=\frac{n}{m}\in \Q$ (with   $n\in \Z,m\in \N\setminus\{0\}$). Then by what we just proved $nf( \one_\Omega)=f(n\one_\Omega)=f(mq\one_\Omega)=mf( q\one_\Omega).$ Thus
$f(q\one_\Omega)=\frac{n}{m}f(\one_\Omega)=qf(\one_\Omega)$.
Finally, let $k\in\R$, then  there are two   sequences  $(q^1_n)_n\subseteq \Q $, $(q^2_n)_n\subseteq \Q $ such that $q^1_n\uparrow k$ and $q^2_n\downarrow k$. By the first part of the proof  and by   monotonicity of $f$, for all $n$, $q^1_nf(\one_\Omega)=f(q^1_n\one_\Omega)\leq f(k\one_\Omega) \leq f(q^2_n\one_\Omega)=q^2_nf(\one_\Omega)$. Letting $n\rightarrow \infty$ shows that $f(k\one_\Omega)=kf(\one_\Omega)$.  

\p\noindent
$\bullet$ [$\Leftarrow$] The proof is immediate noticing that $f( k\one_\Omega)=kf( \one_\Omega)$ for all $k\in \R$.
\end{proof} 

\begin{lemma}\label{continuous} 
Let  $f:B(\Omega,\A)\rightarrow\R$ be monotone and translation invariant, let $k:=f(\one_\Omega)$, then
$$|f(x)-f(y)|\le k\|x-y\|_\infty   \text{ for all}\; x, y \in  B(\Omega,\A).$$
\end{lemma} 
\begin{proof} 
For all $x, y \in  B(\Omega,\A)$ one has $  x   \le y+ \|x-y\|_\infty  \one_\Omega$. Since $f$ is monotone and translation invariant, we deduce that 
$ f( x)   \le f(y+ \|x-y\|_\infty  \one_\Omega)=f(y)+ k\|x-y\|_\infty .$
Exchanging the role of $x$ and $y$, we get $ f( y)   \le  f(x)+ k \|y-x\|_\infty  .$ Thus, $|f(x)-f(y)|\le k\|x-y\|_\infty $. 
\end{proof} 

\begin{lemma}\label{lemma:PCP-TI_and_CM}
A pricing rule $f:B(\Omega,\A)\rightarrow\R$ satisfies PCP if and only if it satisfies Translation Invariance and the following Buy $\&$ Sell Additivity Property:   
\begin{equation}\label{eq:Ch-Dec}
f(x)=f(x\wedge 0) +f(x\vee 0)=f(x^+)+f(-(-x)^+) \;\; \text{ for all } x\in B(\Omega,\A) .
\end{equation}
\end{lemma}
\begin{proof}
$\bullet$ [$\Rightarrow$] First, we have $f(0)=0$ by  PCP, taking   $x=0$ and $k=0$. 

Second, (\ref{eq:Ch-Dec}) follows from PCP,   taking $k=0$, since $f(0)=0$.
 
 We now prove that $f$ is translation invariant. Let $x\in B(\Omega,\A)$ and let  $k\geq 0$. Then   one has: 
\begin{align*}
f(x )& =f(x^+)+f(-(-x)^+) &\text{by (\ref{eq:Ch-Dec})}\\
	&=f((x+ k\one_\Omega-k\one_\Omega)^+)+ f(-(k\one_\Omega-x-k\one_\Omega)^+) \\
	&= f(x+ k\one_\Omega) -f(k\one_\Omega)& \text{by PCP}.
\end{align*}
Consequently, $f$ is translation invariant.

\p\noindent $\bullet$
 [$\Leftarrow$] Fix  $x\in B(\Omega,\A)$ and $k\geq 0$. Then PCP holds since we have:
  \begin{align*}
 f((x-k\one_\Omega)^+)+f(-(k\one_\Omega-x)^+)& = f(x-k\one_\Omega)&\text{ from (\ref{eq:Ch-Dec})}\\
  	&= f(x)-f(k\one_\Omega)  &\text{ from TI and Lemma \ref{lemma:strong_TI}.}
 \end{align*}
 \end{proof}

\begin{lemma}\label{lemma:PCP-ComoAdd}
A monotone pricing rule $f:B(\Omega,\A)\rightarrow\R$ satisfies PCP if and only if it is Comonotonic Additive.
\end{lemma}
\begin{proof}

\noindent{$\bullet$ [$(i)\implies (ii)$]} The following steps prove that $f$ is comonotonic additive.

\begin{step}\label{step:g_finite_como_add}
$f(x+y)=f(x)+f(y)$ for all comonotonic and positive step functions $x,y\in B(\Omega,\A)$.
\end{step}
\begin{proof}
Let $x,y\in B(\Omega,\A)$ be comonotonic and positive step functions. By comonotonicity there is a partition $A_1,\dots,A_n$ of $\Omega$  such that:
\begin{align*}
x &=x_1 \one_{A_1}+\dots+x_n \one_{A_n}, \text{ with } 0\leq x_1\leq \dots\leq x_n, \\
y &=y_1 \one_{A_1}+\dots+y_n \one_{A_n}, \text{ with } 0\leq  y_1\leq \dots\leq y_n. 
\end{align*}
Equivalently, one can write
\begin{align*}
x &= \sum_{i=1}^nX_i, \text{ where }  X_i=(x_i-x_{i-1})\one_{A_i\cup \dots \cup A_n},\ i=1,\dots,n\ \text{ and }\ x_0=0, \\
y &= \sum_{i=1}^nY_i, \text{ where }  Y_i=(y_i-y_{i-1})\one_{A_i\cup \dots \cup A_n},\ i=1,\dots,n\ \text{ and }\ y_0=0. 
\end{align*}

\p\noindent
$\star$ We first show that $f(x)=\sum_{i=1}^n f(X_i)$, $f(y)=\sum_{i=1}^nf(Y_i),\,\, \text{ and }\,\,  f(x+y)=\sum_{i=1}^nf(X_i+Y_i)
$. 

It is enough to prove that for $i=1,\dots,n-1:$
$$f(X_i+X_{i+1}+\dots+X_n)=f(X_i)+f(X_{i+1}+\dots+X_n).$$
For $i$ fixed $1\leq i \leq n,$ set   $\Sigma_{i,n}:=X_{i}+\dots+X_n$. Then we have:
$$\Sigma_{i+1,n}=(x_{i+1}-x_i)\one_{A_{i+1}}+(x_{i+2}-x_i)\one_{A_{i+2}}+\dots+(x_n-x_i)\one_{A_n},$$
$$X_i=\Sigma_{i,n}\wedge (x_i-x_{i-1})\one_{\Omega}=[\Sigma_{i,n}-(x_i-x_{i-1})\one_{\Omega}]\wedge 0+ (x_i-x_{i-1})\one_{\Omega}\ \text{ and }$$
$$\Sigma_{i+1,n}=\Sigma_{i,n}\vee (x_i-x_{i-1})\one_{\Omega}-(x_i-x_{i-1})\one_{\Omega}=[\Sigma_{i,n}- (x_i-x_{i-1})\one_{\Omega}]\vee0.$$
Let us show that $f(\Sigma_{i,n})=f(X_i+\Sigma_{i+1,n})=f(X_i)+f(\Sigma_{i+1,n})$. 
Using Buy $\&$ Sell Additivity (\ref{eq:Ch-Dec}) of $f$ and Lemma \ref{lemma:strong_TI} we get:
\begin{align*}
f(X_i)+f(\Sigma_{i+1,n})&=f([\Sigma_{i,n}-(x_i-x_{i-1})\one_{\Omega}]\wedge 0+ (x_i-x_{i-1})\one_{\Omega})
+f([\Sigma_{i,n}- (x_i-x_{i-1})\one_{\Omega}]\vee0 )\\
 &=f(\Sigma_{i,n}-(x_i-x_{i-1})\one_{\Omega})+f((x_i-x_{i-1})\one_{\Omega})\\
 &=f(\Sigma_{i,n})=f(X_i+\Sigma_{i+1,n}).
  \end{align*}

 By induction it is easy to see that
$
f(x)=\sum_{i=1}^nf(X_i).$ Similarly, we prove that $f(y)=\sum_{i=1}^nf(Y_i),\,\, \text{ and }\,\,  f(x+y)=\sum_{i=1}^nf(X_i+Y_i).
$\qed

\p\noindent
$\star$ To conclude the proof we   show that $f(X_i+Y_i)=f(X_i)+f(Y_i)$ for all $i$. 

Note that $X_i=a\one_A$, $Y_i=b\one_A$ with  $A:= A_i\cup \dots \cup A_n$, $a:=x_i-x_{i-1}\ge 0, b:=y_i-y_{i-1}\ge 0$. Thus we only need to prove that $f(a\one_A+b\one_A)=f(a\one_A)+f(b\one_A).$ Indeed, if $x:= (a+b)\one_A-a\one_\Omega=-a\one_{A^c}+b\one_A$ we have $x\vee 0= b\one_A$, $x\wedge 0= -a\one_{A^c}=a\one_A-a\one_\Omega$. Thus, using Translation Invariance and Buy $\&$ Sell Additivity (\ref{eq:Ch-Dec}) of $f$ we obtain:
\begin{align*}
 f(a\one_A+b\one_A)+f(-a\one_\Omega)&=f(x)=f(x\vee0)+f(x\wedge 0)\\
&=f(b\one_A)+f(a\one_A-a\one_\Omega)=f(b\one_A)+f(a\one_A)+f(-a\one_\Omega). 
\end{align*}
Thus, $f(X_i+Y_i)=f(a\one_A+b\one_A)=f(a\one_A)+f(b\one_A)=f(X_i)+f(Y_i).$
\end{proof}

\begin{step}\label{step:g_pcom}
For all $x,y\in B(\Omega,\A)$ positive and comonotonic, $f(x+y)=f(x)+f(y)$.
\end{step}
\begin{proof}
Consider $x,y\in B(\Omega,\A)$ such that $x\geq0$ and $y\geq0$.  Define for all $n\in\N$
$$
x_n=\sum_{i=0}^{n2^n-1}\frac{i}{2^n}\one_{\left\{\frac{i}{2^n}<x\leq \frac{i+1}{2^n}\right\}}\,\,\text{ and }\,\,
y_n=\sum_{i=0}^{n2^n-1}\frac{i}{2^n}\one_{\left\{\frac{i}{2^n}<y\leq \frac{i+1}{2^n}\right\}}.
$$
Since $x$ and $y$ are bounded above, there exists $N \in \N$ such that $n\geq N$ implies
$$
x_n\leq x\leq x_n+\frac{1}{2^n}\one_{\Omega} \,\,\text{ and }\,\, y_n\leq y\leq y_n+\frac{1}{2^n}\one_{\Omega}.
$$
It is straightforward to check that $x_n$ and $y_n$ are comonotonic since $x$ and $y$ are comonotonic. Therefore Step \ref{step:g_finite_como_add} implies $f(x_n+y_n)=f(x_n)+f(y_n)$. Since $f$ is continuous for the sup norm by Lemma \ref{continuous}, passing to the limit when $n\to \infty$, one gets
$f(x+y)=f(x)+f(y)$
\end{proof}

\begin{step}\label{step:g_com}
For all $x,y\in B(\Omega,\A)$ comonotonic, $f(x+y)=f(x)+f(y)$.
\end{step}
\begin{proof}
Let $x,y\in B(\Omega,\A)$ be comonotonic. We can choose $k\geq0 $ such that $x'=x+k\one_\Omega\geq0$ and $y'=y+k\one_\Omega\geq0$. Since $f$ satisfies TI  
we have: 
$$f(x'+y')=f(x+y+2k\one_\Omega)=f(x+y)+f(2k\one_\Omega)= f(x+y)+2f( k\one_\Omega).$$
By Step~\ref{step:g_pcom}, noticing that $x'$ and $y'$ are comonotonic    and using again TI, we have: 
$$f(x'+y')=f(x')+f(y')=f(x)+f(k\one_\Omega)+f(y)+f(k\one_\Omega).$$
Hence $f(x+y)=f(x)+f(y)$.
\end{proof}

\medskip
\noindent{$\bullet$ [$(ii)\implies (i)$]} Since $f$ is comonotonic additive, it satisfies Buy $\&$ Sell Additivity (\ref{eq:Ch-Dec}) and translation invariance, i.e.:
$$f(x)=f(x^+)+f(-(-x)^+) \;\; \text{ for all } x\in B(\Omega,\A),
$$
$$f(x+k\one_\Omega)=f(x)+f(k\one_\Omega)\,\, \text{ for all } x\in  B(\Omega,\A)\,\, \text{ and all }k\in \R.
$$
This follows from the facts that    
$x^+$ and $-(-x)^+$   are comonotonic  and  $x $ and $k\one_\Omega$ are also comonotonic. Consequently $f$ satisfies PCP by Lemma \ref{lemma:PCP-TI_and_CM}.
\end{proof}

\begin{proof}[Proof of Theorem \ref{th:PCP_Choquet}] 
\medskip

From Lemma \ref{lemma:PCP-ComoAdd},  $f$ satisfies PCP if and only if $f$ is comonotonic additive. From Schmeidler \cite{Schmeidler86}, the comonotonic additivity of $f$ is equivalent to the fact that $f$ is a Choquet pricing rule.
\end{proof}


\subsection{Proofs of Proposition \ref{prop:Sipos-charact}, Theorem \ref{th:CPP_Choquet-Sipos} and Theorem \ref{th:CPP_Choquet-Sipos10}}   \label{proof:CPP_Choquet-Sipos}

We prove first Theorem \ref{th:CPP_Choquet-Sipos10} and we use it to prove Proposition \ref{prop:Sipos-charact} and Theorem \ref{th:CPP_Choquet-Sipos}.

\begin{proof}[Proof of Theorem \ref{th:CPP_Choquet-Sipos10}] 
 $\bullet$ [$(i)\implies (iii)$] Note that Choquet pricing rules are Buy $\&$ Sell Additive, i.e.,  for all $x \in B(\Omega,\A)$,  $\int^C x  \, dv=\int^C x^+  \, dv+\int^C -x^-  \, dv$; indeed, from Schmeidler \cite{Schmeidler86}, every   Choquet integral
 is comonotonic additive and, for all $x \in B(\Omega, \A)$,  $x^+$ and $-x^-$ are comonotonic. Then we have 
\begin{align*}
 \int^C x  \, dv  - \int^S x  \, dv&=    \left[\int^C x^+   \, dv + \int^C -x^-  \, dv\right] -\left[\int^C x^+   \, dv-\int^C x^-  \, dv\right]\\   
&=\int^C x^-  \, dv +  \int^C -x^-  \, dv  =f(x^-)+f(-x^-)\ge 0\; \text{ [by }  (i)].\\
 \int^S x  \, dv  + \int^C -x  \, dv&=  -\int^S -x  \, dv  + \int^C -x  \, dv \ge 0 \; \text{ [from above}].
\end{align*}

\p\noindent  $\bullet$ [$(iii)\implies (i)$] Immediate.

\p\noindent$\bullet$ [$(i)\Rightarrow (ii)$] Fix $A\in \A$ and consider $x:= \one_A$. Then  $(ii)$ holds since:
\begin{align*}
 0 \le     f(-x) +f(x) &= f(\one_{A^c} -\one_\Omega) +f(\one_A) \; \text{[by  ($i$)]} \\
&=f(\one_{A^c}) -f(\one_\Omega) +f(\one_A) \; \text{[from  translation invariance]}\\
&=   v( A^c) - v(\Omega)+v(A). 
\end{align*}

\p\noindent$\bullet$ [$(ii)\Rightarrow (i)$]
Note that we are done as soon as we prove $(i)$ for $x\geq 0$ since $f$ is a Choquet integral and therefore it satisfies TI.  But, from $(ii)$ we have $ v^*\le v$, and  for  $x\geq 0$ we get from standard properties of the Choquet integral
$$
-f(-x)=-\int^C (-x) \, dv = \int^C  x\, dv^*\le \int^C  x\, dv=f(x).
$$
\end{proof} 

\begin{proof}[Proof of Proposition \ref{prop:Sipos-charact}]
 \p\noindent $\bullet$ [$(i)\Leftrightarrow (ii)$] Define the   pricing rule $f^*:  B(\Omega,\A)\rightarrow \R$ by $f^*(x):=-f(-x)$. If $f$ is a Choquet pricing rule w.r.t. $v$, then $f^*$ is a Choquet pricing rule w. r. t. $v^*$. 
 Then Assertion $(i)$ holds if and only if  $f(x)+f(-x)\ge 0$ and $f^*(x)+f^*(-x)\ge 0$ for all $x$, hence  if and only if
   $v\ge v^*$ and $v=(v^*)^*\ge v^*$ by  Theorem \ref{th:CPP_Choquet-Sipos10}. That is, $v=v^*$. 

 \p\noindent $\bullet$ [$(i)\Leftrightarrow (iii)$] If $f$ is   a Choquet-\v Sipo\v s pricing rule, then  $-f(-x)=-\int^S -x \, dv=\int^S  x \, dv =f(x)$ for all $x$. Hence,  $(i)$ is satisfied. Conversely, if  $f$ is a Choquet pricing rule  such that, for all $x$,  $f(x)+f(-x)= 0$, then in particular $f(x)+f(-x)\ge 0$. By Theorem \ref{th:CPP_Choquet-Sipos10},    
 $\int^{S} x \, dv\in [-f(-x), f(x)] =\{ f(x)\}$ for all $x$. Thus $f$ is a \v Sipo\v s pricing rule.
\end{proof}

\begin{proof}[Proof of Theorem \ref{th:CPP_Choquet-Sipos}] 
By Proposition \ref{prop:CCP-PCPandNS},  $f$ satisfies   \ref{eq:CPP} if and only if  $f$ satisfies   \ref{eq:PCP} and $f(x)+f(-x)= 0$ for all $x$.  By Theorem \ref{th:PCP_Choquet}, $f$ satisfies   \ref{eq:PCP}, if and only if $f$ is a Choquet pricing rule. By Proposition \ref{prop:Sipos-charact},  $f$ is a Choquet pricing rule and satisfies $f(x)+f(-x)= 0$ for all $x$ if and only if $f$ is   a Choquet-\v Sipo\v s pricing rule. Thus, $(i)$ is equivalent to $(ii)$. 
 \end{proof} 


\subsection{Proof of Theorem \ref{th:DCP-Sipos}}\label{proof:DCP-Sipos}

Let us denote $B_+(\Omega,\A)=\{x\in B(\Omega,\A) \, | \, x\geq 0\}$ and let TI$_+$ (resp. DCP$_+$) denote TI (resp. DCP$_+$) restricted to $B_+(\Omega,\A)$.

\begin{proof}[Proof of Theorem \ref{th:DCP-Sipos}]
 \p\noindent $\bullet$ [$(i)\Rightarrow (ii)$] By taking $k=0$, DCP implies property \ref{eq:Ch-Dec}, i.e., that $f$ is Buy $\&$ Sell Additive.  Moreover DCP implies TI$_+$. For all $x\in B_+(\Omega,\A)$, $k\ge 0$, apply DCP to $x+ k \one_{\Omega}$ and $k$ to get:
$$
f(x+ k \one_{\Omega})=  f((x+ k \one_{\Omega})\vee k \one_\Omega -k \one_\Omega)+f((x+ k \one_{\Omega})\wedge k \one_\Omega)= f(x) +f( k \one_{\Omega}). 
$$
Doing the same proof as in  Step \ref{step:g_finite_como_add}, Step \ref{step:g_pcom} of Lemma \ref{lemma:PCP-ComoAdd}, we can show that $f$ satisfies Comonotonic Additivity on $B_+(\Omega,\A)$. Then by Schmeidler \cite{Schmeidler86}, $f$ is a Choquet pricing rule on $B_+(\Omega,\A)$.  Taking $k=0$, from DCP (first equality), and NS (third equality)   we deduce that for all $x\in B(\Omega,\A)$
$$
f(x)=f(x\wedge 0) + f(x\vee 0)= f(x^+) + f(-x^-)=f(x^+) -f(x^-)
$$
Since  $x^+,x^-\in B_+(\Omega,\A)$, and since  $f$ is a Choquet integral on $B_+(\Omega,\A)$
$$
f(x)=\int^C x^+ dv - \int^C x^- dv=\int^S x^- dv
$$
i.e., $f$ is a \v Sipo\v s pricing rule.

 \p\noindent $\bullet$ [$(ii)\Rightarrow (i)$] By  Theorem $5(ii)$ in \v Sipo\v s  \cite{Sipos}, the \v Sipo\v s integral is monotone when $v$ is a capacity. It is easy to see that a \v Sipo\v s pricing rule satisfies NS.  Only DCP is left to be shown. Using NS, we get for all $x\in B(\Omega,\A)$
$$
f(x)=f(x^+) -f(x^-)=f(x^+) +f(-x^-)= f(x\vee 0)+ f(x\wedge 0),
$$
i.e.,  $f$ is Buy $\&$ Sell Additive.

We show that $f$ satisfies DCP$_+$. Note that  $(x-k \one_\Omega)^+$ and $x\wedge k \one_\Omega$ are comonotonic. Since $f$ is a \v Sipo\v s integral, it is a Choquet integral on $B_+(\Omega,\A)$ and therefore if satisfies comonotonic additivity   on $B_+(\Omega,\A)$. Using the fact that $ x = (x-k \one_\Omega)^++x\wedge k \one_\Omega$, we get  $ f(x) = f((x-k \one_\Omega)^+) +f(x\wedge k \one_\Omega)$. Using \ref{eq:+_vee_w} we obtain 
$$
f(x\vee 0)= f((x\vee 0)\vee k \one_\Omega -k \one_\Omega)+f((x\vee 0)\wedge k \one_\Omega)=f(x\vee   k \one_\Omega -k \one_\Omega)+f((x\wedge k \one_\Omega)\vee 0)
$$
 since $k\ge 0$ implies  $(x\vee 0)\vee k \one_\Omega =x\vee   k \one_\Omega$ and   
  $(x\vee 0)\wedge k \one_\Omega(\omega)=0=(x\wedge k \one_\Omega)\vee 0(\omega)$ if $x(\omega)\leq 0$ and $(x\vee 0)\wedge k \one_\Omega(\omega)=x\wedge k \one_\Omega(\omega)=(x\wedge k \one_\Omega)\vee 0(\omega)$ if $x(\omega)\geq 0$. Also
 $$
 f(x\wedge 0)= f((x\wedge k \one_\Omega)\wedge0).
 $$
Replacing $f(x\vee 0)$ and $f(x\wedge 0)$ in (\ref{eq:Ch-Dec}) and applying (\ref{eq:Ch-Dec}) once again one obtains
$$
f(x)=f(x\vee   k \one_\Omega -k \one_\Omega)+f((x\wedge k \one_\Omega)\vee 0)+f((x\wedge k \one_\Omega)\wedge0)=f(x\vee   k \one_\Omega -k \one_\Omega)+f(x\wedge k \one_\Omega)
$$
i.e., DCP holds.
\end{proof}



\subsection{Proof of Theorem \ref{th:AF_Choquet}}\label{proof:AF_Choquet}

\begin{proof}
$\bullet$ [Proof of the Choquet Part] $ [\Rightarrow]  $ Assume that $f$ satisfies AF and is a Choquet pricing rule (note that the proof of $\Rightarrow$ works also if $f$ is a \v Sipo\v s pricing and will be needed in the second part). Let $v^*$ be the conjugate of $v$. We   prove that $(i)$ $v^*$ is balanced and $(ii)$ $AC(v)=\core(v^*)$, hence $\core(v^*)\ne \emptyset$  by    Schmeidler \cite{Schmeidler68} that extends  the result by Bondareva \cite{Bondareva1963} and Shapley \cite{Shapley1967} to infinite spaces $\Omega$.\footnote{See also Theorem 5 in Marinacci and  Montrucchio \cite{MarinacciMontrucchio2004}.}

 We first prove that  $v^*$ is balanced. Fix $n\in \N$ and let  $a_1,\dots,a_n\geq0$ and $A_1,\dots,A_n\in \A$ such that  $\sum_{i=1}^n a_i \one_{A_i}-  \one_\Omega  =  0$. Then by AF, $\sum_{i=1}^nf(  a_i\one_{A_i})+f(-  \one_\Omega)\geq 0$. Since $f$ is a Choquet 
 or a \v Sipo\v s
 pricing rule, $f( a_i \one_{A_i})= a_if(\one_{A_i})= a_i v(A_i)$ and $f(- \one_\Omega)=-f( \one_\Omega)=-v(\Omega)$. Therefore, $\sum_{i=1}^n a_i v(A_i)\geq v (\Omega)$ and clearly $\sum_{i=1}^n a_i v^*(A_i)\leq v^*(\Omega)$. Then by Schmeidler \cite{Schmeidler68},  one gets: 
$$
\core(v^*):=\Big\{\mu\in \text{ ba}(\A) : \mu(A)\ge v^*(A) \; \text{ for all } \; A \in\A  \; \text{and } \; \mu(\Omega)= v^*(\Omega) \Big\}\ne \emptyset.
$$
But $v^*$ is a capacity since $v$ is a capacity. Hence all $\mu\in \core(v^*) $ are non-negative. But the set  $ \{\mu \in \text{ ba}(\A):\mu\ge 0, \; \text{and}\;\mu(\Omega)=v(\Omega) \}$ is the set of positive, additive set functions on $\Omega$ such that $\mu(\Omega)=v(\Omega)$. Therefore:
$$
AC(v):= \{\mu:\A \mapsto \mathbb{R}\, :\,\mu \text{ is positive, additive, } \mu\leq v,\text{ and } \mu(\Omega)=v(\Omega)\}=\core(v^*)\ne \emptyset.
$$

\p\noindent $ [\Leftarrow]  $ Assume $AC(v)\ne \emptyset $ and let $\mu \in AC(v)$. Let  $x_1,\dots,x_{n}\in  B(\Omega,\A)$ such that $x:=\sum_{i=1}^nx_i \geq 0$. Then $\int^C  x_i \, dv \geq \int  x_i d\mu $ for all $i$ (since $\mu \le v$ and $\mu(\Omega)  = v(\Omega)$). Since $x\geq 0$, one gets:
$$
\sum_{i=1}^nf(x_i):=\sum_{i=1}^n\int^C  x_i \, dv \geq   \sum_{i=1}^n\int  x_i d\mu  = \int x d\mu\geq 0.
$$
This proves that $f$ is arbitrage free.\qed

\p\noindent
$\bullet$ [Proof of the \v Sipo\v s Part]
If a \v Sipo\v s pricing rule $f$ is linear then it is clearly AF. \\
We now prove the converse implication. Suppose that $f$ is a Sipo\v s pricing rule that satisfies AF. Then,  there exists a    capacity $v$ such that 
$ f(x) =\int^S x\, dv \; \text { for all  } x \in B(\Omega, \A).$ For $f$ to be linear it is sufficient to prove that $v$ is additive since  for $v$ additive one has
$f(x)  =\int^C x^+\, dv \, -\int^C x^-\, dv$ (since f is a \v Sipo\v s  pricing rule) hence 
$f(x) =\int [x^+-x^-]\, dv   = \int x\, dv$ (since  $v$ is additive). We end the proof by showing that $v$ is additive. First note that $ \one_{\Omega}- \one_A- \one_{A^c}=0= \one_A+ \one_{A^c}- \one_{\Omega}$. Since $f$ is AF and $f(-x)=-f(x)$ for all $ x \in B(\Omega, \A)$ (as $f$ is a  \v Sipo\v s  pricing rule), one has:
\begin{align*}
v(\Omega)-v(A)-v(A^c)&= f( \one_{\Omega})+f(- \one_A)+f(- \one_{A^c})\geq 0,\\
v(A)+v(A^c)-v(\Omega) &=f( \one_A)+f( \one_{A^c})+f(- \one_{\Omega})\geq 0.
\end{align*}
 Therefore $v(A)+v(A^c)=v(\Omega)$, that is, $v$ is auto-conjugate.

Finally, from the first part of the theorem,  $AC(v)\neq \emptyset$. Let $\mu \in AC(v)$, then $\mu(A)\leq v(A)$ and $ \mu(A^c)\leq v(A^c)$ for all $A\in \A$. Since  $v$ is auto-conjugate,  neither of the previous inequalities can be strict, otherwise, summing up we would get $v(\Omega)=\mu(\Omega)=\mu(A)+\mu(A^c)< v(A)+v(A^c)=v(\Omega)$, a contradiction. Therefore $\mu(A)= v(A)$  for all $A\in \A$. This proves that $v$ is additive and the Sipo\v s pricing rule $f$ is linear.
\end{proof}


\end{document}